\documentclass[a4, 10 pt]{article}

\usepackage{amsfonts,amssymb, amsmath,amsthm}
\usepackage[utf8]{inputenc}

\usepackage[combine]{ucs}
\usepackage{euscript}
\usepackage{graphicx}
\usepackage{pstricks}
\usepackage{pst-pdf}

\newtheorem{theorem}{\textsc{Theorem}}
\newtheorem{lemma} {\textsc{Lemma}}
\newtheorem{remark}{\textbf{Remark}}
\newtheorem{proposition} {\textsc{Proposition}}

\newtheorem{definition}{\textsc{Definition}}

\newcommand{\E}{\EuScript{E}}

\newcommand{\G}{\EuScript{G}}

\newcommand{\D}{\mathcal{D}}
\newcommand{\I}{\mathbb{I}}

\def \Graph{\mathcal{G}}
\def \grv{V_\mathcal{G}}
\def \gre{E_{\mathcal{G}}}
\def \deg{\mathsf{deg}}
\def \Neigh{\mathsf{N}}
\def \temp{\varkappa}
\def\conf{\eta}
\def\concrit{\eta_{\mathsf{cr}}}
\def\Sand{\mathcal{S}}
\def\bound{v_\mathsf{bound}}
\def\Con{\mathsf{Conf}}
\def\erase{\mathcal{E}}
\def\topple{\mathsf{topple}}
\def\id{\mathsf{Id}}

\begin{document}
\title{Notes on identical configurations in Abelian Sandpile Model with initial height.}
\author{Arnold M.D.
\thanks{International Institute of the Earthquake Prediction theory RAS}
\thanks{Kharkevich's Institute of Information Transition Problems RAS}}
\maketitle
\begin{abstract}
The aim of this note is to systematize our knowledge about identical configurations of ASM.
\end{abstract}
\section{Introduction.}
Abelian sandpile model (ASM) was introduced  by Bak, Tang and Wiessenfeld in their work \cite{BTW} describing formation of avalanches. In most general formulation the model can be defined as following automata. Let $\Graph=(\grv, \gre)$ denote a finite graph. For any vertex $v\in\grv$ denote by 
$\Neigh(v)$ the set of all adjacent vertices $\Neigh(v)=\{v_j\in \grv\mid (v,v_j)\in \gre\}$ and by  
$\deg(v)=|\Neigh(v)|$ the degree of $v$. 

Fix some positive integer parameter $\temp$. Integral-valued function $\conf: \grv\mapsto \{\temp +\mathbb{N} \}$ is called a configuration on graph $\Graph$ with potential $\temp$.

Sandpile transformation $\Sand$ acts on the space of  configurations $\Con(\Graph)$ by two steps:

\begin{enumerate}
\item Increase value of $\conf(v_0)\mapsto \conf(v_0)+1$ for randomly chosen vertex $v_0\in\grv$. 

\item If an updated value of $\conf$ at some vertex $v'$ exceeds its \emph{critical} value $\temp+d(v')$ \emph{topple} $\conf$ at $v'$ i.e.
\begin{itemize}
\item $\conf(v')\mapsto \conf(v')-d(v')$

\item$\conf(v)\mapsto \conf(v)+1$ for all $v\in\Neigh(v')$
\end{itemize}
Such relaxation process may be written in the form
\[\mathsf{topple}(\conf)=\conf-\Delta \mathbb{I}(\conf(v)>\temp+d(v))\]

\end{enumerate} 

It is natural to set number $d(v)$ to be equal $\deg(v)$ so that total norm of the configuration will not change during the toppling procedure.
However in this case relaxation process described above will never stop for configuration $\conf(v)=(\temp+\deg(v))+\delta_{v_0}$ and thus Sandpile transformation will be ill-defined. natural way to avoid this is to define a set of \emph{boundary} vertices $\partial \Graph=\{\bound\}$ were toppling will decrease the configuration $\conf(\bound)$ by some number $d(\bound)>\deg(\bound)$ and thus total weight of the configuration $\|\conf\|_{L^1}=\sum\limits_{v\in\grv} \conf(v)$ will dissipate through $\partial\Graph$.

Original situation considered in \cite{BTW} provides highly illustrative example. Let $\Graph$ be a bounded subset of two-dimensional lattice $\mathbb{Z}^2$. Every internal node has exactly four neighbours and its degree is also $4$. On the other hand any boundary node has strictly less then four neighbours. For $\Con=4^{\grv}$ toppling of the node always decrease the value of configuration by $4$ and so every boundary node dissipate the total weight of configuration each time toppling process goes through it (see Fig \ref{fig: boundary_toppling}).

\begin{figure}[h]
\centering
\begin{tabular}{cc}
\includegraphics[width=1 in]{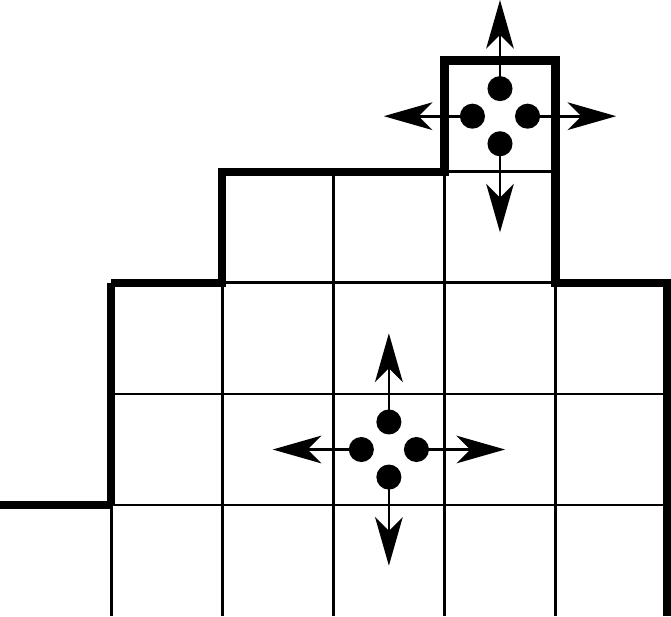}&
\includegraphics[width=2.5 in]{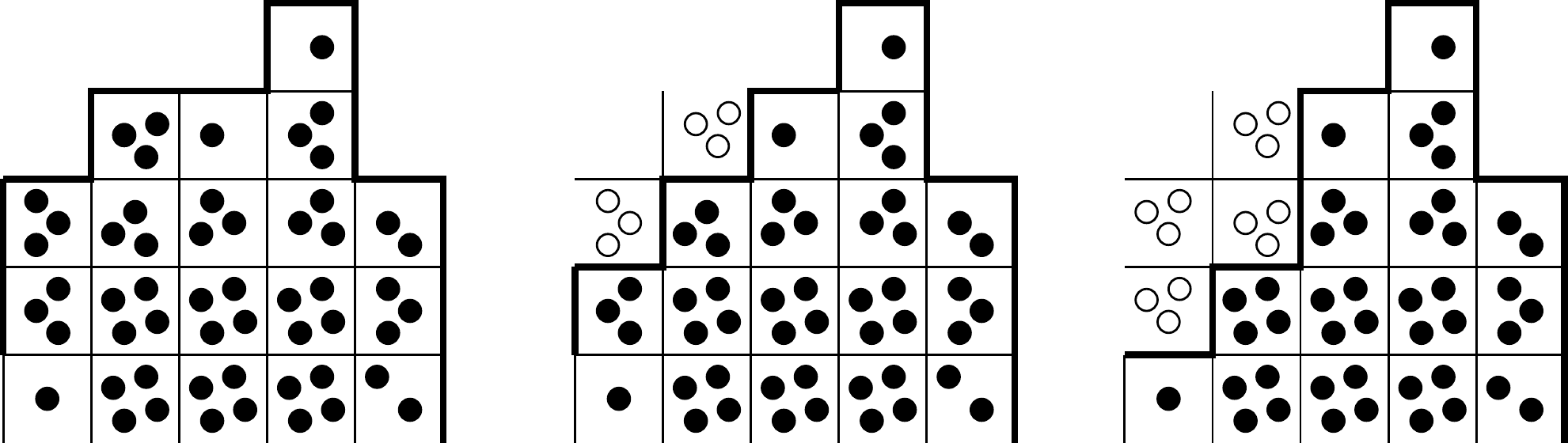}\\
a)&b)
\end{tabular}
\caption{a) Boundary nodes has less then four neighbours. b) Burning test}
\label{fig: boundary_toppling}
\end{figure}

\begin{theorem}[(see \cite{Jarai})]
\label{th: well_def}
Sandpile transformation is well-defined for each finite graph $\Graph$ with boundary: $\Sand_\xi\conf$ depends only on initial configuration $\conf$ and vertex $v_\xi$ and does not depend on the sequence in which toppling procedure were done.  
\end{theorem}

Transformation $\Sand$ being very non-local in the sense of Hausdorf metric $|\conf-\conf'|_\mathcal{H}=|\{v\in\grv: \conf(v)\ne\conf'(v)\}|$ on the space of configurations, defines thanks to theorem \ref{th: well_def} Markov process on $\Con$ with very remarkable properties.

High interest to this Markov process was caused by the critical behaviour of the distribution of the quantities $|\Sand\conf-\conf|_{\mathcal{H}}$ and $|\Sand\conf|_1-|\conf|_1$. Set of recurrent states for the process $\Sand$ was thus very intensively studied over the past decades. In this section we state some theorems which can be found in  \cite{LevineSpanning}, \cite{DharAlgebraic},\cite{DharRelated},\cite{DharSOC}\cite{LevineExplosions}, \cite{Jarai} and  references therein describing the structure of this set.

\begin{definition}  
\label{def: erasable} Let $\conf$ be a configuration on $\Graph$ with spin $\temp$.

\begin{itemize}
\item Vertex $v\in\grv$ is called $0$--erasable for  configuration $\conf$ if $\conf(v)\geqslant \temp+\deg(v)$.
Set of all $0$--erasable vertices for configuration $\conf$ is denoted by $\erase_0(\conf)$

\item Vertex $v\in\grv$ is called $j$--erasable for  configuration $\conf$ if \[\conf(x)\geqslant \temp+\deg(\left.v\right|_{\grv\setminus\bigsqcup\limits_{k=0}^{j-1}
\erase_k(\conf)})\]
Set of all $j$--erasable vertices for configuration $\conf$ is denoted by $\erase_j(\conf)$
 
\item Configuration $\conf$ is called erasable if there exists such $N$ that $\grv=\bigsqcup\limits_{j=0}^{N}\erase_j(\conf)$.
Set of all erasable configurations is denoted by 
$\E(\temp,\Graph)$
\end{itemize}
\end{definition}

\begin{theorem}[see \cite{Jarai}]
\label{recurrent_erasable}
Set $\E(\temp, \Graph)$ coincides with the set of recurrent configurations of the Sandpile process.
\end{theorem}

From definition \ref{def: erasable} one can easily notice that

\begin{remark}[Monotonicity 1]
If $\conf\in\E(\temp,\Graph)$ and $\conf'>\conf$ then $\conf'\in\E(\temp,\Graph)$.
\end{remark}

\begin{definition}
\label{def: Graph_embedding}
We shall say that graph $\Graph_1$ is embedded in $\Graph_2$ in the sense that $\grv^{1}\subset \grv^{2}$, $\gre^1\subset\gre^2$ and for each $v\in\grv^1$: $\deg_2(v\mid_{\grv^1})=\deg_2(v)$. 
\end{definition}

Using again the classical graph for Sandpile model we shall say that subset $\Omega_1\subset\mathbb{Z}^2$ is embedded into $\Omega_2\subset\mathbb{Z}^2$ if $\Omega_1\subset\Omega_2$.

 From definition it immediately follows that

\begin{remark}[Monotonicity 2] For $\conf\in\E(\varkappa,\Graph_2)$ the restriction  $\conf|_{\Graph_1}\in\E(\varkappa,\Graph_1)$.
\end{remark}

One of the most remarkable properties of the set $\E(\temp,\Graph)$ is presented in the next theorem

\begin{theorem}[see \cite{LevineSpanning}]
\label{th: erasable_spanning}
Set $\E(\temp,\Graph)$ is bijective to the set of all spanning trees on $\Graph$.
\end{theorem}

There exists natural bijection $\E(\temp,\Graph)\mapsto \E(\temp+1,\Graph)$. Namely
\begin{equation}
\label{eq: bijection}
\conf\in \E(\temp,\Graph)\Leftrightarrow \conf+\overline{\mathbf{1}}\in\E(\temp+1,\Graph)
\end{equation}
where $\overline{\mathbf{1}}$ denotes a function $\overline{\mathbf{1}}: v\in\grv\mapsto 1$. Thus one can consider $\E(\temp,\Graph)$ for one value of $\temp$. Unfortunately, bijection \eqref{eq: bijection} does not hold algebraic structure of the set $\E(\temp,\Graph)$ which was noticed in fundamental paper \cite{DharAlgebraic}.

\begin{theorem}[see \cite{DharAlgebraic}]
\label{th: isomorphism_Erasable}
Set $\E(\temp,\Graph)$ with the operation $\conf,\,\conf'\rightarrow \conf\conf':=\mathsf{topple}(\conf+\conf')$ is isomorphic to the set $\mathcal{F}/_{\Delta}$ of equivalence classes of functions on $\Graph$ up to the image of the Laplace operator. Such a factor space has a structure of Abelian group.  
\end{theorem}

In this paper we are mainly interested in the \emph{identical} configuration, i.e. erasable configuration which belongs to the class of equivalence of $\{0\}$. In the section \ref{sec: theory} we present some theoretical results concerning identical configurations on graphs. Section \ref{sec: Numerics} will be dedicated to experimental results. In section \ref{sec: Sierp} we describe identical configuration for the Sierpinskii graph. At last in section \ref{sec: Open}  we provide a proof of an upper bound of $|\Sand \conf -\conf|_{\mathcal{H}}$ on $\mathbb{Z}^2$ and pose some open questions.

Some results on critical behaviour of Abelian Sandpile Model (ASM) on $\mathbb{Z}^2$ can be found in: \cite{ASM_Chaos}, \cite{Redig_ASM}, \cite{DharSOC}, \cite{BTW}, \cite{Jarai}, \cite{Priezzev} and references therein.

Connection between ASM and similar models on $\mathbb{Z}^2$ is observed in: \cite{Dhar_pattern}, \cite{LevineExplosions}, \cite{LevinePeres_Rotor}, \cite{Dhar_Pattern_multipple},\cite{Redig_Heap}.

Neutral configurations of ASM on $\mathbb{Z}^2$ were addressed by Creutz in \cite{Creutz} and were studied in \cite{Iden_Paoletti}, \cite{Iden_Rossin}, \cite{Iden_Lattice}.

ASM on other graphs such as Sierpinski graph and other self-similar fractal structures were studied in \cite{Dhar_Sierp}, \cite{Sand_Seirp1},\cite{Sand_Sierp2}, \cite{Sierp_Bengal},\cite{Sierp_prop}, \cite{Spanning_Sierp}, \cite{Nagnibeda_Basilica}, \cite{NagnibedaASM}.

\paragraph*{Acknowledgements.} 
Author is deeply thankfull to E.I. Dinaburg and A.N. Rybko for fruitfull discussions. 

\section{Some theory.}
\label{sec: theory}

Theorem \ref{th: isomorphism_Erasable} leads to the definition 
\begin{definition}
Define Green function for two erasable configurations as follows
\[\Delta G^{(\conf,\conf')}=\conf+\conf-\conf\conf'\]
\end{definition}

\begin{theorem}(see \cite{DharAlgebraic})
\label{th: Green-function}
$G^{(\conf,\conf')}(v)$ equals the number of topplings occurred at $v$ in the process of relaxation of the element $\conf+\conf'\mapsto \conf\conf'$. 
\end{theorem}

\begin{proof} Compute the number of incoming and outcoming particles at vertex $v$ in the relaxation process.
Total income has the form
\[\sum\limits_{v'\in \Neigh(v)} G^{(\conf,\conf')}(v)-\deg(v)G^{(\conf,\conf')}(v).\]
\end{proof}

Now we are able to notice some properties of the function $G^{(\conf,\conf')}$.

\begin{proposition}[Monotonicity I]
\label{th: monotony_configurations}
If $\conf\geqslant \conf'$ then for any $h\in\E(\temp,\Graph)$ it follows that $G^{(\conf,h)}\geqslant G^{(\conf',h)}$.
\end{proposition}

\begin{proof} Use Theorem \ref{th: Green-function}. If $\conf>\conf'$ then relaxation of $h+\conf$ can be considered as two consecutive relaxations thanks to Abelian property.
\[\topple(h+\conf)=\topple((\conf-\conf')+\topple(h+\conf'))\]
 \end{proof} 

\begin{remark} Simple computation yields 
\[\Delta G^{(\conf,h)}=
\conf+h-\conf h=\conf'+h+(\conf-\conf')-\conf h=\]
\[=\conf'+h-\conf'h+(\conf-\conf')-(\conf h-\conf'h)=\Delta G^{(\conf',h)}+(\conf-\conf')-(\conf h-\conf' h)\]
thus 
\begin{equation}
\conf h-\conf'h\leqslant \conf-\conf'
\end{equation} for any $\conf,\conf',h\in \E(\temp, \Graph)$.
\end{remark}

\begin{proposition}
\label{th: Green_recurrence}
For any $\conf,\conf'$ and $h$ \[G^{(\conf,h)}-G^{(\conf',h)}=G^{(\conf,\conf'h)}-G^{(\conf h,\conf')}\]
\end{proposition}

\begin{proof} Goes from the definition
\[\conf \conf' h=\conf+\conf'h-\Delta G^{(\conf,\conf'h)}=\conf+\conf'+h-\Delta (G^{(\conf',h)}-G^{(\conf,\conf'h)})=\]\[fh+\conf'+\Delta (G^{(\conf h,\conf')}-G^{(\conf',h)}-G^{(\conf,\conf'h)})=\conf h\conf'+\Delta (G^{(\conf h,\conf')}+G^{(\conf,h)}-G^{(\conf',h)}-G^{(\conf,\conf'h)})\]

Thanks to Dirichlet boundary conditions the only harmonic function is identically zero. Which yields the result. 

\end{proof}

\begin{proposition}
\label{th: Green_boundary}
For any $\conf, \conf'$ 
\[\sum\limits_{v\in\grv} \Delta G^{(\conf,\conf')}(v)=\sum\limits_{\bound} \deg(\bound)G^{(\conf,\conf')}(\bound)\]
\end{proposition}

\begin{proof} Compute the particles which topples out of the boundary. \end{proof}

\begin{definition}
Unique configuration $\id \in \E(\temp,\Graph)$ such that for any other configuration $\conf\in\E(\temp,\Graph)$
\[\id \eta=\eta\]
is called identical configuration.
\end{definition}
Existence and uniqueness of such configuration is granted by Theorem \ref{th: isomorphism_Erasable}.

Denote by $\conf^*$ maximal configuration. 
\[\conf^*(v)=\temp+\deg(v)\]

\begin{theorem}
\label{th: minimax}
\[\min\limits_\conf\max\limits_{\conf'} G^{(\conf,\conf')}=G^{(\id,\conf^*)}=:G_\temp\]
\end{theorem}

\begin{proof} Proof goes in two steps. First, by theorem \ref{th: monotony_configurations} 
$\max\limits_\conf' G^{(\conf,\conf')}=G^{\conf,\conf^*}$.

At second, since $\conf^*\geqslant h$ for any $h\in \E(\temp,\Graph)$ then by theorem \ref{th: monotony_configurations} we get
\[G^{(\conf,\conf^*)}\geqslant G^{(\conf,\id)}\]
Since $\id\in \{0\}$ it means that $\id=\Delta G_\temp$ for some $G_\temp\in \mathbb{Z}^\grv$. Thus $G^{(\conf,\id)}=G_\temp$ for any $\conf\in \E(\temp,\Graph)$. In particular, 
\[G^{(\id,\conf^*)}=G_\varkappa\]
\end{proof}

\begin{proposition}[Monotonicity II]
\label{th: monotonicity_area}
If $\Graph_1$ is embedded in $\Graph_2$ then \[G_\temp^{\Graph_2}\geqslant G_\temp^{\Graph_1}\]
\end{proposition}

\begin{proof}  Let $\id'$ denote the restriction of identical configuration $\id\in \E(\temp,\Graph_2)$ on $\Graph_2$ to the set $\grv^1$. 
Then $G^{(\conf^*,\id')}\leqslant G_\temp\mid_{\grv^1}$. Denote \[h=\conf^*-\id'\conf^*\] 

Then $G^{\conf^*,h}=G^{\conf^*,\id'}$ and $G^{\conf^*,h}\geqslant G_\temp^{\Graph_1}$ since 
obviously $\conf^*+h\in \{\conf^*\}$.
\end{proof}

\section{Identity on Sierpinski carpet.}
\label{sec: Sierp}

As it was mentioned in the introduction, there is a natural bijection \eqref{eq: bijection} between two sets of erasable configurations with different spins. Unfortunately, in general $\overline{\mathbf{1}}$ doesn't belong to $\{0\}$ and so bijection \eqref{eq: bijection} isn't isomorphic. Thus question about identical configuration is the question about the orbit of function $\overline{\mathbf{1}}$. 

\begin{proposition}
For any $n\in\mathbb{N}$ there exists such $\temp$ that $\id_\temp\in\{(\overline{\mathbf{1}})^n\}$.
\end{proposition}

\begin{proof} Since every object under consideration is finite there exists a cycle in the sequence $\{(\overline{\mathbf{1}})^k\}_{k\in\mathbb{N}}$. Thus the set $\{(\overline{\mathbf{1}})^k\}_{k\in\mathbb{N}}$ forms a subgroup in the set of all erasable configurations $\E(\temp)$. For instance there exists such $n$ that $f\in\{\overline{\mathbf{1}}^n\}\cap \E(\temp)=\{\overline{\mathbf{1}}^{-1}\}$. Thus  applying bijection \eqref{eq: bijection} one gets $f+\overline{\mathbf{1}}\in\{0\}\cap \E(\temp+1)$. 
\end{proof}

We present the following table illustrating, the fact, that such orbit can be sufficiently large.
Consider $\Omega$ containing only three consequent cells.
 
\begin{center}
\begin{tabular}[l]{|c|c|c|}
\hline
$x$&$y$&$z$\\
\hline
\end{tabular}
\end{center}

It follows from symmetry that for any $\temp$ $\id(x)=\id(z)$ so one can get.

\begin{center}
\begin{tabular}[l]{|c||c|c||c|c|}
\hline
$\temp$&$\id(x)$&$\id(y)$&$G(x)$&$G(y)$\\
\hline
$14q+1$&$2$&$1$&$5q+1$&$6q+1$\\
\hline
$14q+2$&$1$&$0$&$5q+1$&$6q+1$\\
\hline
$14q+3$&$3$&$1$&$5q+2$&$6q+2$\\
\hline
$14q+4$&$2$&$0$&$5q+2$&$6q+2$\\
\hline
$14q+5$&$0$&$3$&$5q+2$&$6q+3$\\
\hline
$14q+6$&$3$&$0$&$5q+3$&$6q+3$\\
\hline
$14q+7$&$1$&$3$&$5q+3$&$6q+4$\\
\hline
$14q+8$&$0$&$2$&$5q+3$&$6q+4$\\
\hline
$14q+9$&$2$&$3$&$5q+4$&$6q+5$\\
\hline
$14q+10$&$1$&$2$&$5q+4$&$6q+5$\\
\hline
$14q+11$&$3$&$3$&$5q+5$&$6q+6$\\
\hline
$14q+12$&$2$&$2$&$5q+5$&$6q+6$\\
\hline
$14q+13$&$1$&$1$&$5q+5$&$6q+6$\\
\hline
$14q$&$3$&$2$&$5q+1$&$6q+1$\\
\hline
\end{tabular}
\end{center}

\begin{remark} There are cases with two possible configurations in the table. For $\varkappa=14q+11$ one get $|0|1|0|$ which is unerasable and $|3|3|3|$ which is erasable. Similarly, for $\varkappa=14q$ one get unerasable $|0|0|0|$ and erasable $|3|2|3|$.
\end{remark}

Question about orbit of particular element of the group is very interesting and can be addressed to the future research. We shall not cover it in this survey (see section \ref{sec: Open} ). 

Thus situations when $\overline{\mathbf{1}}\in \{0\}$ are somehow exceptional since in that case question about identical configuration makes sense.

In this section we shall consider another well-known regular graph of order $4$ - Sierpinski carpet. 

The only argument to consider such a fractal here is the following 

\begin{theorem}
On $N$--th Sierpinski carpet identical configuration has the form
\[
\id(x)\equiv \temp+3 \qquad  \mbox{for} \quad\temp= 2n+1
\]
\end{theorem} 

However for $\temp=2n$ identical configuration does not equal to constant (see Fig \ref{fig: serp}).

\begin{figure}
\label{fig: serp}
\begin{center}
\includegraphics[width=3 in, angle=90]{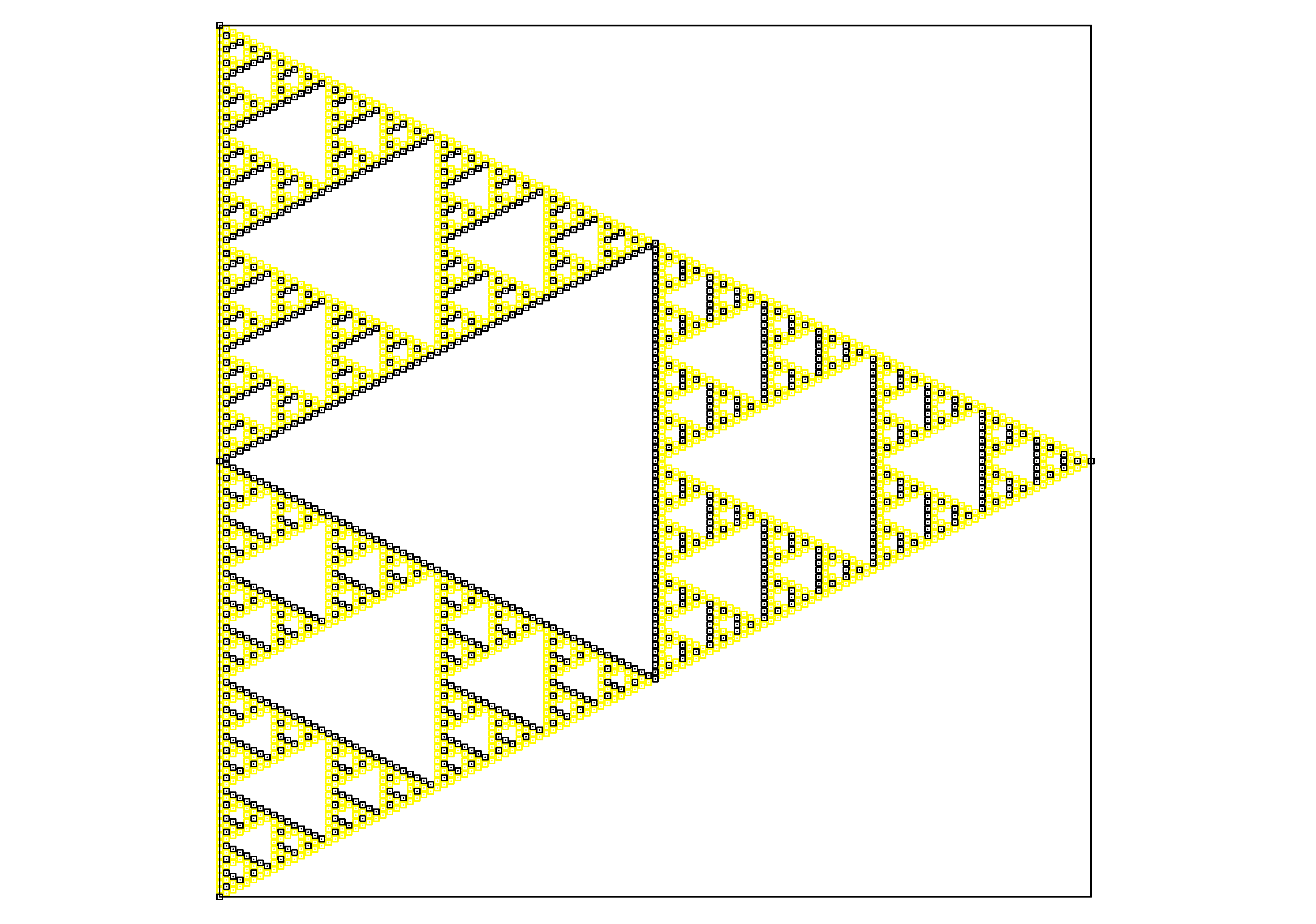}
\end{center}
\caption{Identical configuration on Sierpinski graph for $\temp=0$. }
\end{figure}
We hope that there exists an elegant proof of this fact different from ours which is just constructible. We shall prove the existence of the function $\G_\varkappa$ corresponding to the identical configuration.

Since all sites in $\id$ has the same value $4$ one can easily compute the value of $G$ at the boundary.
While toppling sum of two identical configurations energy can dissipate only through the boundary points with the rate equal $4-\deg(\bound)=4-2$. From symmetry it follows that energy will dissipate from all three vertices equally. Thus one should only calculate the number of particles in identical configuration, which is an easy task. Thus
\begin{equation}
\label{eq: G(a)}
G(\bound)=\frac 13 \cdot \frac{1}{4-2}\|\id\|_1=\frac 13 \cdot \frac{1}{4-2}\left(
4\cdot 3\cdot \left(1+\frac{3^{N+1}-1}{2}\right)\right)=3^{N+1}+1
\end{equation}

Secondly we shall prove reduction lemma
\begin{lemma}
\label{lm:reduction_H}
If $4G(v_0)=H_0+\sum\limits_{j=1}^4G(a_j)$ where $a_j\in \mathbb{S}^{n}$ then 
$4G(x_0)=5H_0+\sum\limits_{j=1}^4G(b_j)$ with $b_j\in \mathbb{S}^{n-1}$. 
\end{lemma}

\begin{proof} Proof consists in straight calculation:

\begin{figure}[ht]
\begin{center}
\begin{pspicture}(0,0)(5,5)

\psline[linewidth=1 pt,linecolor=black,linestyle=solid](1,0)(5,0)
\psline[linewidth=1 pt,linecolor=black,linestyle=solid](1,0)(3,3.5)
\psline[linewidth=1 pt,linecolor=black,linestyle=solid](3,3.5)(5,0)

\psline[linewidth=1 pt,linecolor=black,linestyle=solid](3,0)(4,1.75)
\psline[linewidth=1 pt,linecolor=black,linestyle=solid](3,0)(2,1.75)
\psline[linewidth=1 pt,linecolor=black,linestyle=solid](2,1.75)(4,1.75)

\psline[linewidth=1 pt,linecolor=black,linestyle=solid](2,0)(1.5,0.875)
\psline[linewidth=1 pt,linecolor=black,linestyle=solid](2,0)(2.5,0.875)
\psline[linewidth=1 pt,linecolor=black,linestyle=solid](1.5,0.875)(2.5,0.875)

\psline[linewidth=1 pt,linecolor=black,linestyle=solid](4,0)(4.5,0.875)
\psline[linewidth=1 pt,linecolor=black,linestyle=solid](4,0)(3.5,0.875)
\psline[linewidth=1 pt,linecolor=black,linestyle=solid](3.5,0.875)(4.5,0.875)

\uput[12](2.7,-0.4){$v_0$}

\uput[12](1.7,-0.4){$a_1$}
\uput[12](3.7,-0.4){$a_2$}
\uput[12](0.7,-0.4){$b_1$}
\uput[12](4.7,-0.4){$b_2$}

\uput[12](1.45, 1.8){$b_4$}
\uput[12](3.9,1.8){$b_3$}
\uput[12](2.33,0.9){$a_4$}
\uput[12](2.88,0.9){$a_3$}

\uput[12](4.4,0.9){$a_6$}
\uput[12](0.86,0.9){$a_5$}

\end{pspicture}
\end{center}
\label{fig: reduction}
\caption{Reduction lemma}
\end{figure}

\begin{eqnarray}
4G(v_0)=H_0+G(a_1)+G(a_2)+G(a_3)+G(a_4)\label{eq: x_0}
\\
4G(a_1)=H_0+G(v_0)+G(a_4)+G(a_5)+G(b_1)\label{eq: a_1}
\\
4G(a_2)=H_0+G(v_0)+G(a_3)+G(a_6)+G(b_2)\label{eq: a_2}
\\
4G(a_3)=H_0+G(v_0)+G(a_2)+G(a_6)+G(b_3)\label{eq: a_3}
\\
4G(a_4)=H_0+G(v_0)+G(a_1)+G(a_5)+G(b_4)\label{eq: a_4}
\\
4G(a_5)=H_0+G(a_1)+G(a_4)+G(b_1)+G(b_4)\label{eq: a_5}
\\
4G(a_6)=H_0+G(a_2)+G(a_3)+G(b_2)+G(b_3)\label{eq: a_6}
\end{eqnarray}

Sum over \eqref{eq: a_1}--\eqref{eq: a_4} and introduce \eqref{eq: x_0}

\begin{equation}
\label{eq: reduction_1}
2\sum\limits_{j=1}^4G(a_j)=5H_0+\sum\limits_{j=1}^4G(b_j)+2(G(a_5)+G(a_6))
\end{equation}

Sum separately \eqref{eq: a_5} and \eqref{eq: a_6}

\begin{equation}
\label{eq: reduction_2}
4(G(a_5)+G(a_6))=2H_0+\sum\limits_{j=1}^4G(a_j)+\sum\limits_{j=1}^4G(b_j)
\end{equation}

Multiply \eqref{eq: reduction_1} by $2$ and introduce \eqref{eq: reduction_2}

\[
4\sum\limits_{j=1}^4G(a_j)=10H_0+2\sum\limits_{j=1}^4G(b_j)+2H_0+
\sum\limits_{j=1}^4G(a_j)+\sum\limits_{j=1}^4G(b_j)
\]  

or
\begin{equation}
\label{eq: reduction_3}
\sum\limits_{j=1}^4G(a_j)=4H_0+\sum\limits_{j=1}^4G(b_j)
\end{equation}  

Introducing \eqref{eq: reduction_3} into \eqref{eq: x_0} we obtain
\begin{equation}
\label{eq: reduction_0}
4G(v_0)=5H_0+\sum\limits_{j=1}^4G(b_j)
\end{equation}  
\end{proof}

From reduction lemma it follows that 
\begin{equation}
\label{eq: G(x_0)}
G(v_0)=G(a)+\frac 12 5^n H_0
\end{equation} 

Denote by $H_n:=5^n H_0$ and by $M_n:=G(v_{N-n})-G(a)$ then
\begin{equation}
\label{eq: M_N}
M_N=\frac 12 H_N
\end{equation}

\begin{lemma}
\label{th: Reduction_M}
\[M_{n-1}-H_{n-1}=\frac 35 M_n\]
\end{lemma}

From lemma \ref{th: Reduction_M} it follows

\begin{equation}
\label{eq: M_n}
M_n=\frac{2\cdot 3^{N-n}-1}{2}H_n=(4\cdot 3^{N-n}-2)5^n
\end{equation}

Finally to reconstruct the function $G$ on the whole graph $\mathbb{S}^N$ we use the following lemma.

\begin{lemma}
For given values $G(a)$, $G(b)$ and $G(c)$ for $a$, $b$, $c\in \mathbb{S}^n$ value at the point $v_a\in \mathbb{S}^{n+1}$ which lies on the side opposite to $a$ is equal to 
\[G(v_a)=\dfrac 15(G(a)+2G(b)+2G(c)+3H_n)\]   
\end{lemma}

\begin{proof}
Using the same trick from lemma \ref{lm:reduction_H} we can write
\[G(v_a)+G(v_b)+G(v_c)=G(a)+G(b)+G(c)+2H_n\]
and so
\[4G(v_a)=G(b)+G(c)+G(v_b)+G(v_c)+H_n=G(a)+
2G(b)+2G(c)+3H_n -G(v_a)
\]
\end{proof}

\section{One particular result.}
\label{sec: Open}

Here we shall consider one particular case of $\Omega$ and provide one locality result which can be useful for construction of limiting dynamics. 

Denote by $\D_R$ the diamond of radius $R$
\[\D_R=\{x\mid |x|_{\mathsf{Man}}:=|x_1|+|x_2|\leqslant R\}\]
%
%

%
%
\begin{theorem}[Locality]
\label{th:  Main result one particle}
For any fixed $r$ and sufficiently large $R$ for any $\conf\in \E(\varkappa,\D_R)$ such that $\mathsf{supp}(\concrit-f\}\subseteq \D_r$ and for any point $x\in \D_r$

\[\mathsf{supp}\{(\delta_x\concrit)-\delta_x \conf)\}\subseteq \D_{r+3}\cup \{x_1=0\}\cup\{x_2=0\}\]  
\end{theorem}

\subsection{Proof of theorem \ref{th:  Main result one particle}.}

First we deduce the statement of the theorem from some pure constructive proposition and after that we shall present proofs of that propositions.

\begin{proposition}
\label{th: mozhno perenosit}
For any $f$ from the statement of the theorem and for any point $x$ 
there exists such configuration $f'$ that $\delta_x f=\delta_0 f'$. 
\end{proposition}

So we can consider the most general case $x=0$.
Denote by $\tilde{f}$
\[
\tilde{f}(x)=\left\{ \begin{tabular}{ll}
$\varkappa+2d-3$, &$|x|_1=r+2,\: |x_1x_2|\ne 0$\\
$\varkappa+2d-2$, &$|x|_1=r+2,\: |x_1x_2|=0$\\ 
$f(x)$, & else
\end{tabular}
\right.
\]

\begin{proposition}
\label{th: ovrag mozhno ryt}
$f\in \E(\D_R)\Rightarrow \tilde{f}\in \E(\D_R)$. 
\end{proposition}

\begin{proof}
Since all of the points in $\D_R\setminus \D_{r+2}$ are obviously erasable, it is sufficient to show that $\tilde{f}\in \E(\varkappa, \D_{r+2})$.
We will show even stronger result that every point in the belt $\D_{r+2}\setminus\D_r$ is erasable independently of the configuration $f$ (since it differs from $\overline{\varkappa}$ only in $\D_r$). Then since $f\in\E(\varkappa,\D_r)$ the statement will be proven.

Points $(\pm(r+2),0)$, $(0,\pm(r+2))$ are $0$--erasable by definition. Points $(\pm(r+1),0)$, $(0,\pm(r+1))$ are then $1$-erasable, since their value is $2d-1+\varkappa$. Points $(\pm(r+1),\pm 1)$, $(\pm 1,\pm(r+1))$ are $2$-erasable and so $(\pm r,\pm 1)$, $(\pm 1,\pm r)$ are $3$ - erasable and so on.

In general points $(x_1,x_2)$, $|x|_1=r+2$ are $(2\min(|x_1|,|x_2|))$-erasable and their inner neighbours $(x_1-1,x_2)$ and $(x_1,x_2-1)$ are $(2\min(|x_1|,|x_2|)-1)$-erasable and $(2\min(|x_1|,|x_2|)+1)$-erasable consequently. 
\end{proof}

Thus one can write 
\begin{equation}
\label{eq: ovrag}
f= \delta_{(\pm(r+2),0)}\delta_{(0,\pm(r+2))}\prod_{y_1y_2\ne 0\atop{|y|_1=r+2}}\delta_y^2 \tilde{f}\end{equation}
and all operations $\delta_{x}$ in \eqref{eq: ovrag} commute so for $\delta_0 f$ we get from \eqref{eq: ovrag}

\begin{equation}
\label{eq: zasypat ovrag}
\delta_0 f= \delta_{(\pm(r+2),0)}\delta_{(0,\pm(r+2))}\prod_{y_1y_2\ne 0\atop{|y|_1=r+2}}\delta_y^2(\delta_0 \tilde{f})
\end{equation}
 
It is easy to check that 
\[ \mathrm{supp} G^{(\delta_0, \widetilde{\varkappa+2d-1})}_{\D_R}\subseteq\D_{r+2}\]
Then from theorem \ref{th: monotony_configurations} one can conclude that nothing topples out of the region $\D_{r+1}$ for the configurations $\tilde{f}$ for any erasable configuration $f$. 

If for any $x:\:|x|_1=r+2$ we have $\delta_0 \tilde{f}(x)=\tilde{f}(x)$ then from \eqref{eq: zasypat ovrag} $\delta_0 f(x)<\varkappa+2d-1$ and nothing topples out of the $\D_{r+2}$ so the statement of theorem is satisfied. 

Else there are some points $x$ on the boundary such that
 $\delta_0 \tilde{f}(x)>\tilde{f}(x)$ so after returning excavated particles they should topple. 
In other words for some $x$ one will get

\[\delta_0 \tilde{f}(x)+\delta_{(\pm(r+2),0)}+\delta_{(0,\pm (r+2))}+2\sum\limits_{|y|_1=r+2}\delta_y>\varkappa+2d-1\]

We shall carefully follow the process of toppling and prove by induction, that any site $x$ of the configuration topples not more than  $R-|x|_1+1$ times.

We shall distinguish two kinds of toppling 

\begin{enumerate}
\item Toppling in the domain $\D_R\setminus\D_{r+1}$ 
\item Toppling inside $\D_{r+1}$
\end{enumerate}

\paragraph*{Toppling of the first kind.}

\begin{lemma}
\label{th: connected}
For any connected set $M$ and any point $x^{(0)}\in \partial M$ 
\[G^{\delta_{x^{(0)}},\overline{\varkappa}}_\varkappa=1\]
\end{lemma}

\begin{proof}

If the set $M$ is connected then for any point $x\in M$ there exists a path from $x^{(0)}$ to $x$.  Obviously, if this path contains only cells with $2d-1+\varkappa$ particles and the starting point of this path topples then each cell should topple. So \[G^{\delta_{x^{(0)}},\overline{\varkappa}}(x)\geqslant 1\]

The aim is to prove that there is an identity. The proof goes by induction of the area of $M$. For $M$ containing only one cell the statement is obvious. Suppose that the lemma is proven for any connected set $M$ consisting of $N$ cells.


Consider such set $M'=M\cup\{x'\}$ that $x^{(0)}\in\partial M'$.

Since $M'$ is connected then  $x'$ has not more than $2d$ neighbours from $M$.  By the induction statement any of this neighbours toppled precisely one time and so, since $x'\notin M$ each of them became less or equal than $\varkappa+2d-2$ after such toppling since they have less than $2d$ neighbours. 

Cell $x'$ receive not more than $2d$ particles and so topples one time and distribute $2d$ particles between its neighbours. So any neighbour gets $1$ particle and became not more than $2d-1+\varkappa$. 
\end{proof} 

\begin{proposition}
\label{th: connected_belt}
For any belt $M=\D_{R_1}\setminus\D_{R_2}$ such that $2<R_1-R_2$  and for any point $x: |x|_1=R_1$ 
\[G_M^{\delta_x,\overline{\varkappa}}(x)=1\]
and \[\delta_x \overline{\varkappa}=\overline{\varkappa}-2\sum\limits_{|y|_1=R_1\atop{|y_1y_2|\ne 0}} \delta_y-\delta_{(\pm(R_1),0)}-\delta_{(0,\pm(R_1))}-2\sum\limits_{|y|_1=R_2\atop{|y_1y_2|\ne 0}} \delta_y+\delta_{(\pm(R_2),0)}+\delta_{(0,\pm(R_2))}\]
\end{proposition}

\begin{proof}

For $R_1-R_2>2$ belt $M$ is a connected belt so, by the lemma \ref{th: connected} $G^{(\delta_{x_0},\overline{\varkappa}_M)}=1$.

Any cell which receive as much particles as much neighbours it have and lose $2d$ particles. It means that any cell which does not belong to the boundary receive and lose equal amount of particles, so it remains $2d-1+\varkappa$. Cells at the outer boundary lose $2+\delta(\pm(R_1),0)+\delta(0,\pm(R_1))$ particles. At last cells on the inner boundary lose $2-\delta(\pm(R_2),0)-\delta(0,\pm(R_2))$ particles.
\end{proof}

Thus for one toppling of the first kind each cell $x\in\partial\D_{r+2}$ gets $\eta_x(\D_R\setminus \D_{r+2})$ particles. Clearly $\eta_x(\D_R\setminus \D_{r+2})=2+\delta_{(\pm(r+2),0)}+\delta_{(0,\pm(r+2))}$. 

\paragraph*{Toppling of the second kind.}
\begin{lemma}
\label{th: erasable_toppling}
For any set $M$ define a function
\[\I(x)=\left\{\begin{array}{cc}
2d-\eta_x(M),& x\in \partial M\\
0,& else  
\end{array}\right.\]

Then for any $f\in\E(\varkappa, M)$
\[G_M^{(\I, f)}=1\]
\end{lemma}

\begin{proof}
%
%
%
%
%
%
%
%

Obviously, $\I(x)=\Delta \overline{1}(x)$.
\end{proof}

Now calculate the number of particles which any cell $(x,y)\in\partial\D_{r+1}$ gets during one toppling of the second kind. It gets $\eta_x(\D_{r+1})$ particles. In other words any cell at $\partial\D_{r+1}$ receives $2-\delta(\pm(r+1),0)-\delta(0,\pm(r+1))$ particles. So the number of particles in each point except outer boundary stays unchanged after one step of toppling of the first and second kind and so some points at $\partial \D_{r+2}$ remains greater than $4$. 

Now we can consider only $\D_{R-1}$ instead of $\D_R$ and go another step of induction. and while toppling of the first kind any cell on the boundary $\partial \D_R$ receive $2-\delta(\pm(R),0)-\delta(0,\pm(R))$ particles. So values on the edges remain unchanged and values in the vertices become $\varkappa+2d-3$. 
This circumstance finish the proof of the theorem.

\subsection{Some conjectures and open questions}
\label{sec: open}
\begin{itemize}
\item There are several questions arising from \eqref{eq: bijection}. How does the period of the cycle $\overline{\mathbf{1}}^n$ depend on the set $\grv$? How does this set distributed in the whole set $\E$? What can we say about asymptotic behaviour of the "dimensionless" function $\dfrac{G_\temp}{\temp}$? 

One can conjecture that for the sets consisting only of their boundary such an orbit contains all "symmetric" configurations.
Thus for example periods of this orbit for the few first subsets of two-dimensional lattice are:

\begin{center}
\[\begin{array}{clclcl}
\begin{tabular}[l]{|c| }
\hline
$\phantom{1}$ \\
\hline
\end{tabular} &$T=4$,
&
\begin{tabular}[l]{|c|c| }
\hline
$\phantom{1}$ &$\phantom{1}$\\
\hline
\end{tabular} & $T=3$,
&
\begin{tabular}[l]{|c|c|c| }
\hline
$\phantom{1}$ &$\phantom{1}$&$\phantom{1}$\\
\hline
\end{tabular} & $T=14$,
\end{array}
\]
\end{center} 

\begin{center}
\[\begin{array}{clclcl}
\begin{tabular}[l]{|c|c| }
\hline
$\phantom{1}$ &$\phantom{1}$\\
\hline
$\phantom{1}$ &$\phantom{1}$\\
\hline
\end{tabular} & $T=2$,
&
\begin{tabular}[l]{|c|c|c|c| }
\hline
$\phantom{1}$ &$\phantom{1}$&$\phantom{1}$ &$\phantom{1}$\\
\hline
\end{tabular} & $T=11$,&

\begin{array}{c}
\begin{tabular}[l]{|c| }
\hline
$\phantom{1}$\\
\hline
\end{tabular}
\\
\begin{tabular}[l]{|c|c|c| }
\hline
$\phantom{1}$ &$\phantom{1}$&$\phantom{1}$\\
\hline
\end{tabular}
\end{array} & $T=13$
\end{array}
\]
\end{center} 

However this is certainly not true for the general case. Thus for the square $3\times 3$ such an orbit contains only $16$ configurations.

\item Another set of questions comes from the correspondence to the spanning trees model. Is it right, that the longest tree corresponds to the smallest possible erasable configuration? If it is so, then the minimal weight of erasable configuration is asymptotically $2+\temp$ which somehow correlates with the weight of identical configuration.

\item What can one say about the mean level of the cell in $\grv$ over all spanning trees? (Cesaro mean)?

\item How does the structure of graph affects the structure of the orbit $\overline{\mathbf{1}}^n$ and thus how it is related to the criticality?    
\end{itemize}

 \section{Numerical Experiments}
\label{sec: Numerics}
Numerical experiments provide the evidence of some remarkable properties of the identity configuration on $\mathbb{Z}^2$. Being non-invariant under the change $\temp\to \temp+1$ they still preserve their internal structure and self-similar portraits. We claim that such a rigidity is caused by the underlying rigidity of the corresponding functions $G_\temp$ and thus study of these functions, presented in the section \ref{sec: theory} might be of some interest.
\begin{figure}[htb]
\begin{center}
\begin{tabular}{ccc}
\includegraphics[width=2 in]{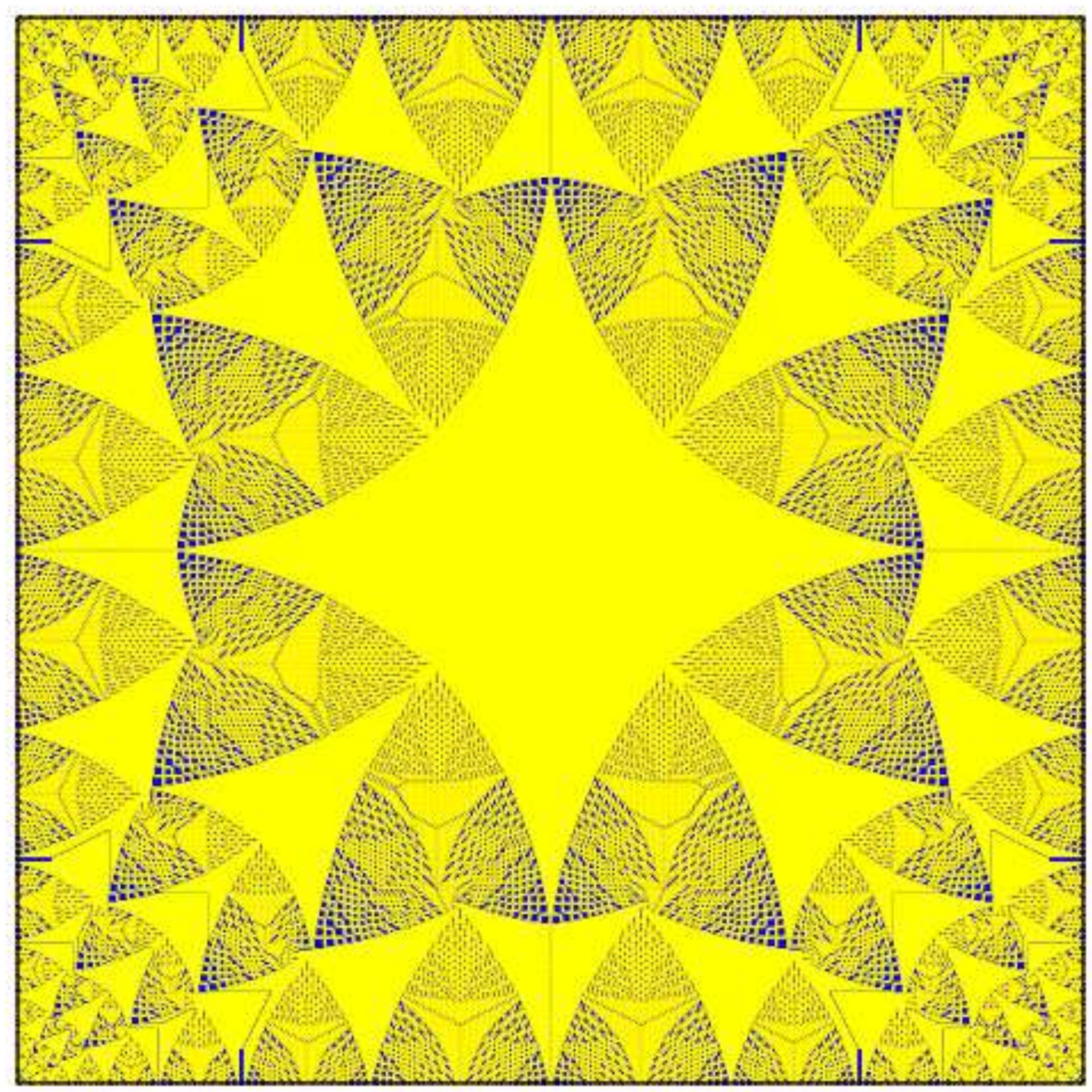}
&
\includegraphics[width=2 in]{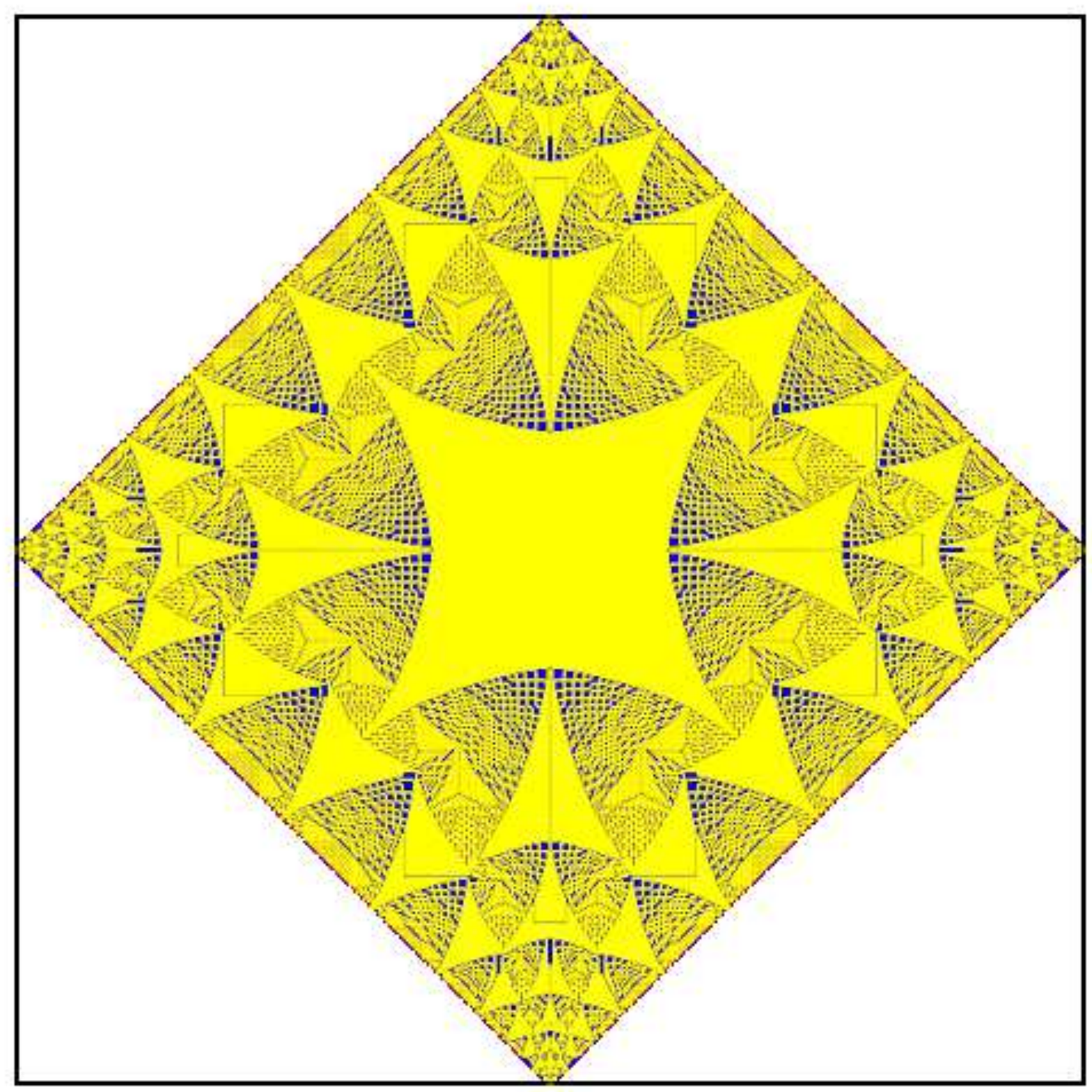}
&
\includegraphics[width=2 in]{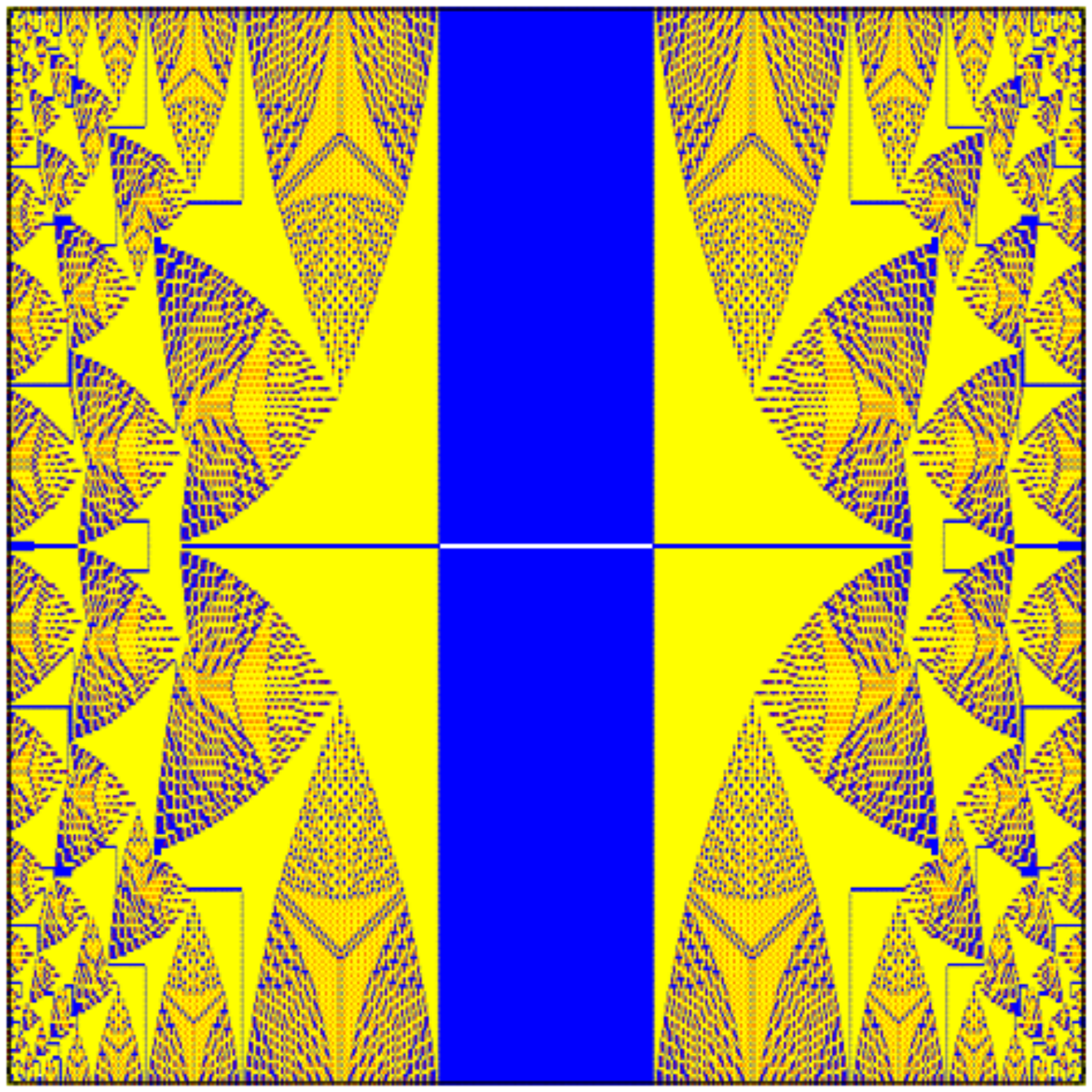}
\\
\includegraphics[width=2 in]{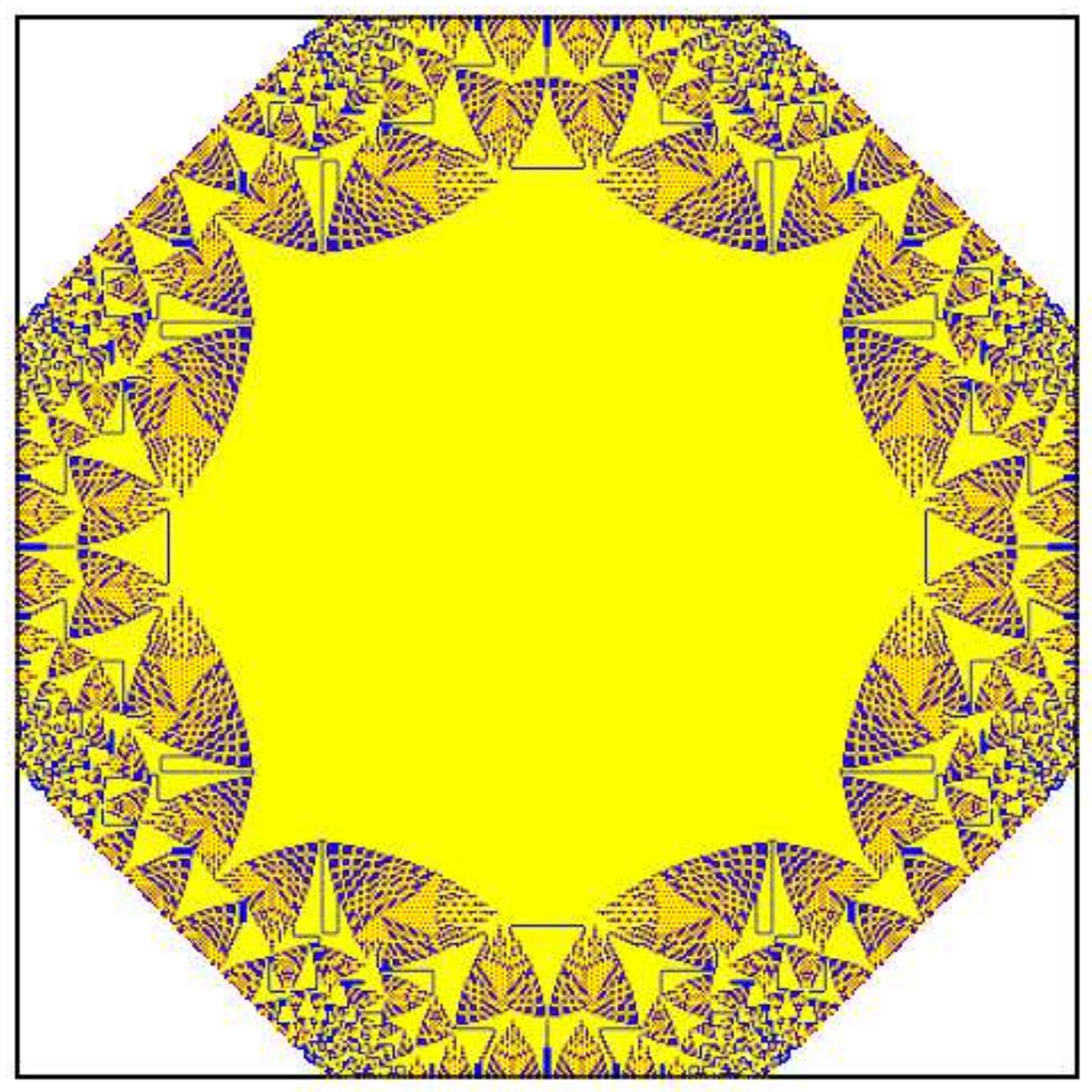}
&\includegraphics[width=2 in]{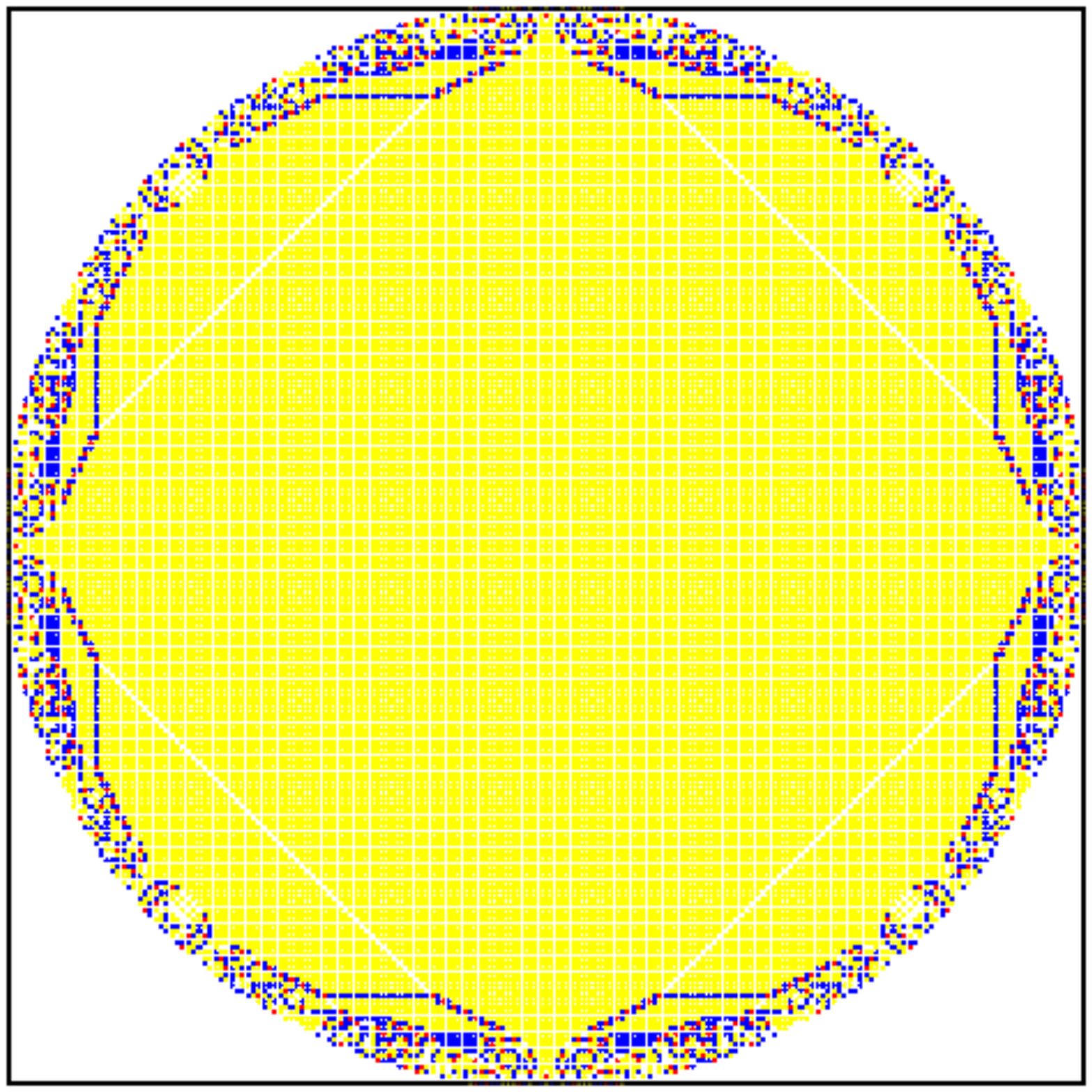}
&\includegraphics[width=2 in]{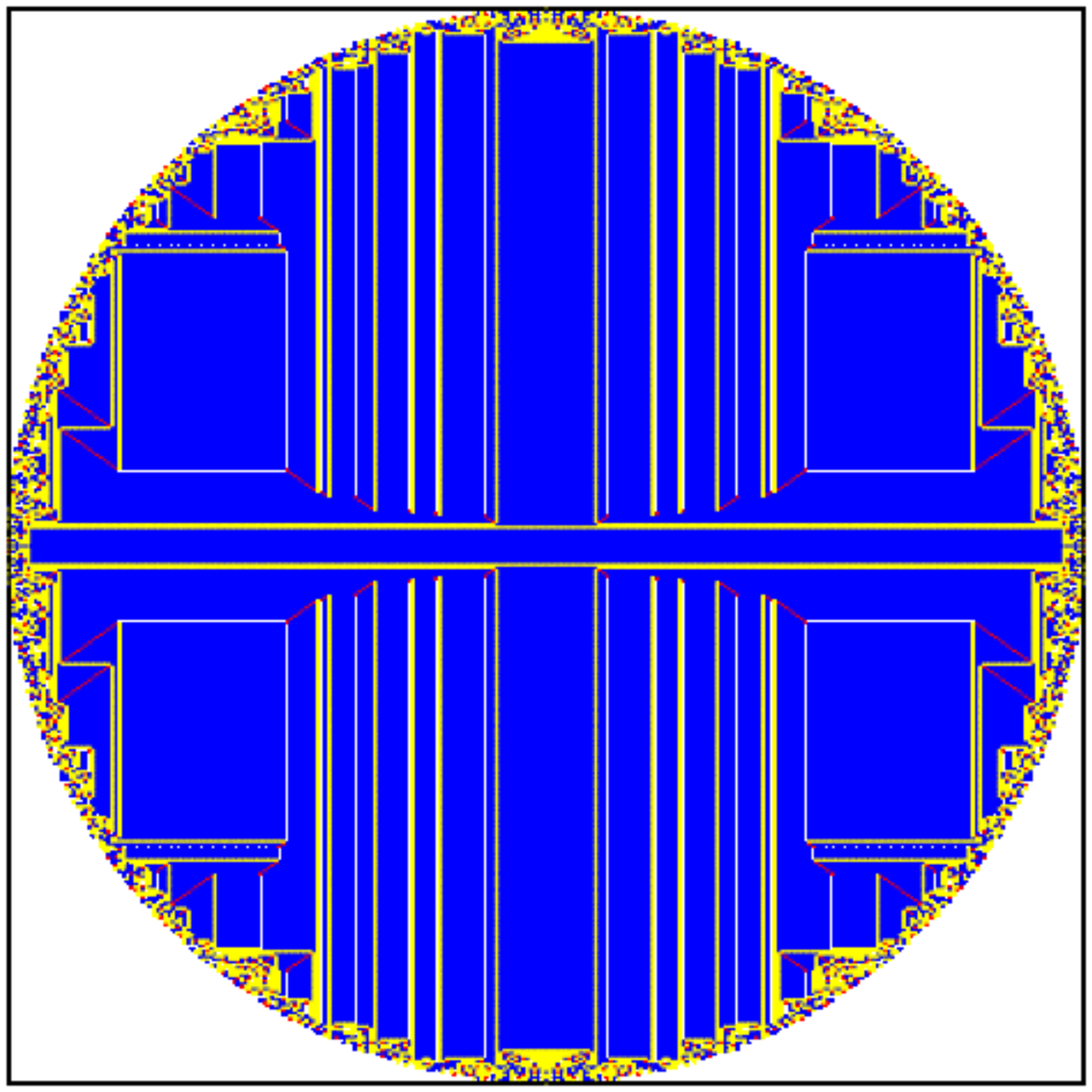}
\end{tabular}
\end{center}
\caption{Identity configuration for $\temp=1$ seems to be rotational invariant and self similar. Top raw: identical configurations for square, diamond and rectangle for $\temp=1$. Yellow points correspond to $\conf(v)=4$, blue to $3$, red -- $2$ and black -- $1$. Bottom raw: identical configurations for octagon, circle and ellipse with eccentricity $0.5$.}
\end{figure}

\begin{figure}[htb]
\begin{center}
\begin{tabular}{ccc}
\includegraphics[width=2 in]{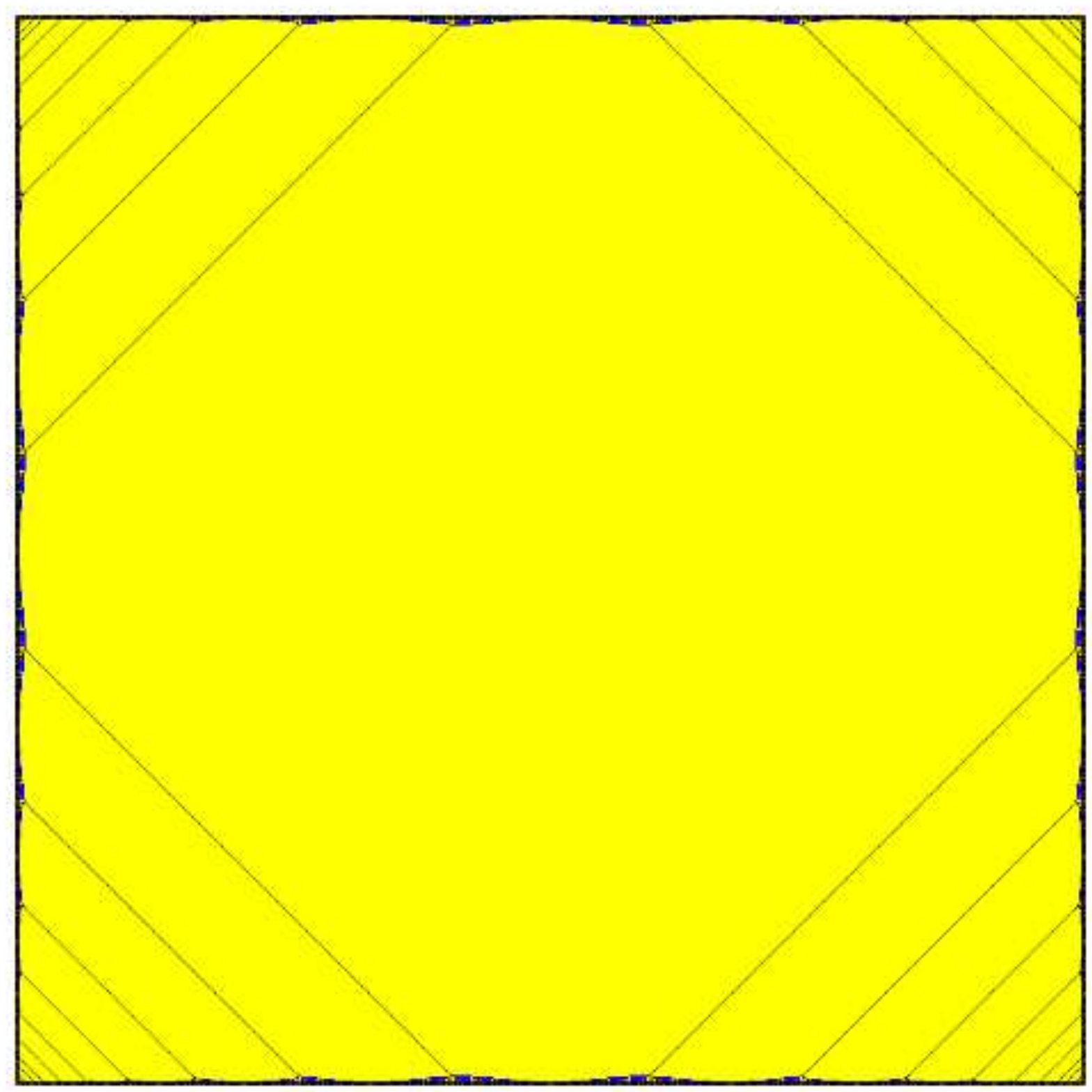}
&
\includegraphics[width = 2in]{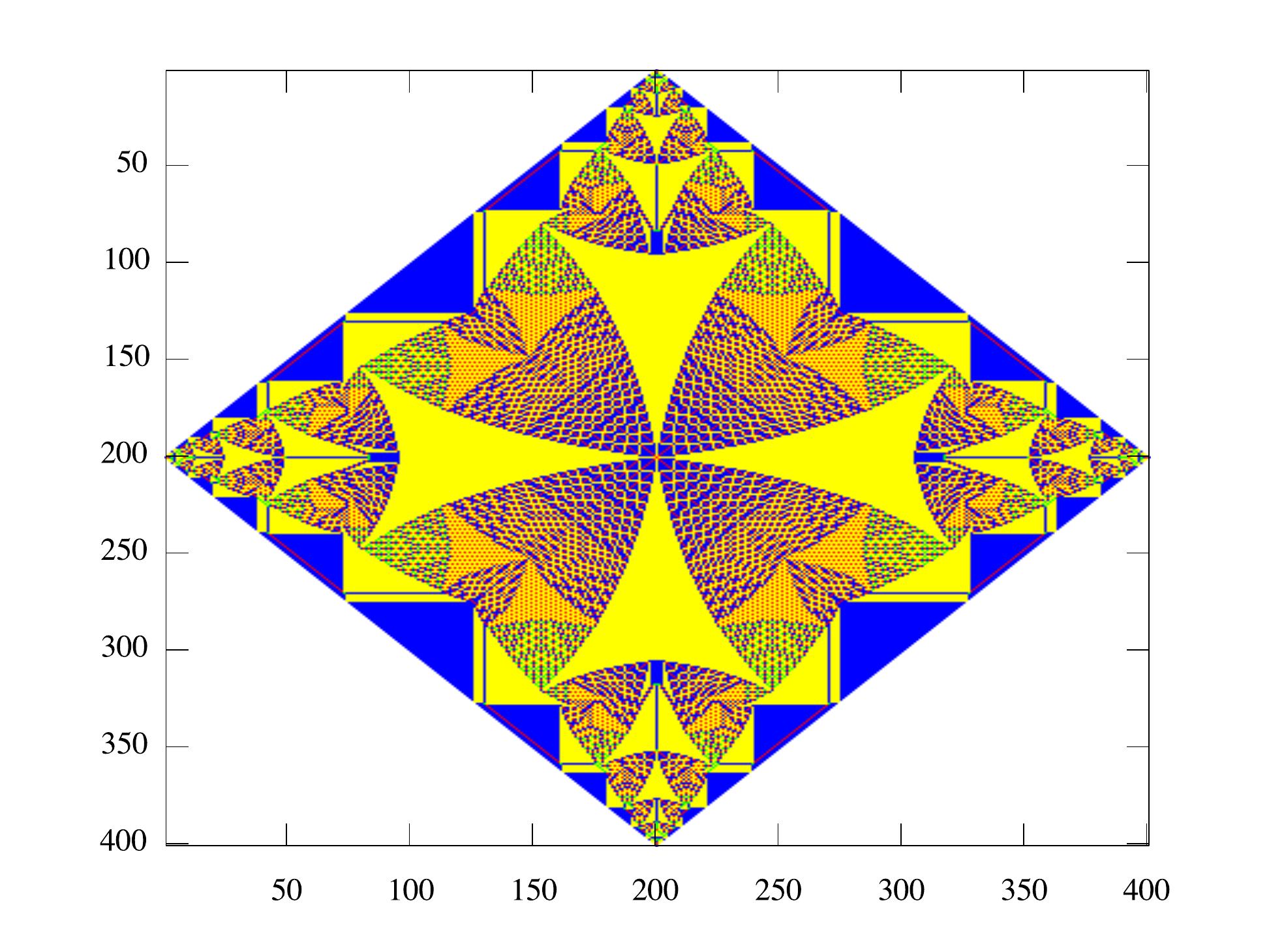}
&
\includegraphics[width = 2in]{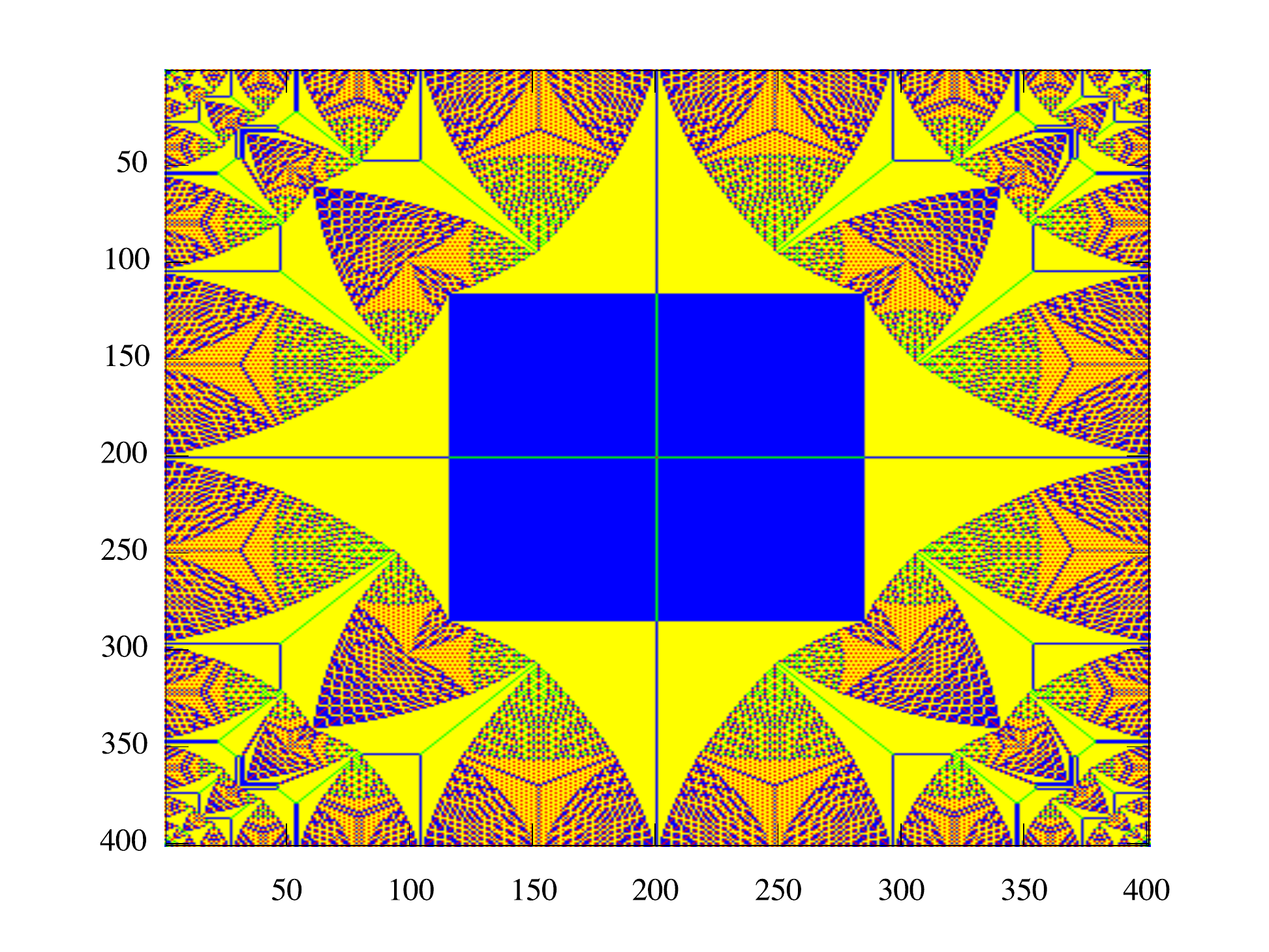}
\end{tabular}
\end{center}
\caption{Left to right: Critical configuration for $\temp=1$ on the square $605\times 605$ with added identical configuration for the square $603\times 603$. Identical configurations for the square and diamond for $\temp=0$.}
\end{figure}

\bibliographystyle{plain}
\bibliography{Sandpile_notes}

\begin{thebibliography}{10}

\bibitem{Iden_Lattice}
C.~Magnien A.~Dartois.
\newblock Results and conjectures on the sandpile identity on a lattice.
\newblock {\em Discrete Mathematics and Theoretical Computer Science}, 2003.

\bibitem{Redig_Heap}
F.~Redig A.~Fey-den Boer.
\newblock Limiting shapes for deterministic centrally seeded growth models.

\bibitem{Redig_ASM}
F.~Redig A.~Fey-den Boer.
\newblock Organized versus self-organized criticality in the abelian sandpile
  model.

\bibitem{Sierp_prop}
Anant P.~Godbole Alberto M.~Teguia.
\newblock Sierpin ́ski gasket graphs and some of their properties.

\bibitem{LevineExplosions}
Yuval~Peres Anne~Fey, Lionel~Levine.
\newblock Growth rates and explosions in sandpiles.
\newblock {\em J. Stat. Phys.}, 2010.

\bibitem{Sierp_Bengal}
Bengal.
\newblock http://www.math.cornell.edu/~bengal/page2b.html.
\newblock {\em J. Stat. Phys.}, 2010.

\bibitem{Sand_Sierp2}
Sava Miloˇsevi ́c2 H. Eugene~Stanley1 Brigita Kutjnak-Urbanc1,
  Stefano~Zapperi1.
\newblock Sandpile model on sierpinski gasket fractal.

\bibitem{Creutz}
M.~Creutz.
\newblock Cellular automata and self organized criticality.
\newblock In {\em Int. conf. on multi- scale phenomena and their simulation}.

\bibitem{Dhar_Sierp}
J.~Lykke Jacobsen D.~Dhar D.~Das, S.~Dey.
\newblock Critical behavior of loops and biconnected clusters on fractals of
  dimension d < 2.

\bibitem{Dhar_pattern}
S.~Chandra D.~Dhar, T.~Sadhu.
\newblock Pattern formation in growing sandpiles.

\bibitem{DharAlgebraic}
S.~Sen D.-N.~Verma D.~Dhar, P.~Ruelle.
\newblock Algebraic aspects of abelian sandpile models.
\newblock {\em J. Phys. A}, 1995.

\bibitem{Nagnibeda_Basilica}
M.~Matter T.~Nagnibeda D.~D’Angeli, A.~Donno.
\newblock Schreier graphs of the basilica group.

\bibitem{Priezzev}
P.~Grassberger V. B.~Priezzhev D.~V.~Ktitarev, S.~L\"ubeck.
\newblock Scaling of waves in the bak-tang- wiesenfeld sandpile model.
\newblock {\em Physical Review}.

\bibitem{DharRelated}
Deepak Dhar.
\newblock The abelian sandpile and related models.
\newblock {\em Physica A}, 1999.

\bibitem{DharSOC}
Deepak Dhar.
\newblock Studying self-organized criticality with exactly solved models.
\newblock {\em arXiv:cond-mat/9909009v1}, 1999.

\bibitem{Sand_Seirp1}
C.~Vanderzande F.~Daerden.
\newblock Sandpiles on a sierpinski gasket.

\bibitem{Jarai}
A.~Jarai.
\newblock Abelian sandpiles: an overview and results on certain transitive
  graphs.

\bibitem{LevineSpanning}
Lionel Levine.
\newblock Sandpile groups and spanning trees of directed line graphs.
\newblock {\em Journal of Combinatorial Theory A.}, 2010.

\bibitem{LevinePeres_Rotor}
Yuval~Peres Lionel~Levine.
\newblock Strong spherical asymptotics for rotor-router aggregation and the
  divisible sandpile.
\newblock {\em Potential Analysis}, 2009.

\bibitem{NagnibedaASM}
T.~Nagnibeda M.~Matter.
\newblock Abelian sandpile model on randomly rooted graphs and self-similar
  groups.
\newblock 2010.

\bibitem{BTW}
C.~Tang P.~Bak and K.~Wiesenfeld.
\newblock Self-organised criticality.
\newblock {\em Phys. Rev. A}, 1988.

\bibitem{Spanning_Sierp}
F.~Redig.
\newblock Mathematical aspects of the abelian sandpile model.
\newblock In {\em Mathematical statistical physics, Les Houches Summer School}.

\bibitem{Iden_Rossin}
D.~Rossin.
\newblock {\em Proprietes Combinatoires de Certaines Familles d’Automates
  Celluuaries}.
\newblock PhD thesis, E\'cole Polytechnique, 2000.

\bibitem{Iden_Paoletti}
A.~Sportiello S.~Caracciolo, G.~Paoletti.
\newblock Explicit characterization of the identity configuration in an abelian
  sandpile model.
\newblock {\em J.Phys.A,}, 2008.

\bibitem{ASM_Chaos}
A.~Mollabashi S.~Moghimi-Araghi.
\newblock Chaos in sandpile models.

\bibitem{Dhar_Pattern_multipple}
D.~Dhar T.Sadhu.
\newblock Pattern formation in growing sandpiles with multiple sources or
  sinks.

\end{thebibliography}
\end{document}